\documentclass[10pt,twoside]{siamltex}
\usepackage{amsfonts,epsfig}

\newtheorem{example}[theorem]{Example}
\newtheorem{problem}[theorem]{Problem}

\def\IR{{\bf R}}
\def\IC{{\bf C}}

\def\diag{{\rm diag}\,}
\def\rank{{\rm rank}\,}
\def\tr{{\rm tr}\,}

\def\vec{{\rm vec}}

\def\qed{\hfill $\Box$\medskip}
\def\diag{{\rm diag}\,}

\def\IR{{\bf R}}
\def\IC{{\bf C}}

\def\cS{{\cal S}}
\def\cR{{\cal R}}
\def\tr{{\rm tr}}

\def\IC{{\mathbb C}}
\def\IR{{\mathbb R}}

\def\cR{{\mathcal R}}
\def\cS{{\mathcal S}}

\def\diag{{\rm diag}\,}

\def\tr{{\rm tr}}

\def\rank{{\rm rank}\,}
\def\diag{{\rm diag}}
\def\[{\left [}
\def\]{\right ]}
\def\({\left (}
\def\){\right )}

\def\Ra{{\ \Rightarrow\ }}

\def\dfrac{\displaystyle\frac}

\begin{document}
\openup .1\jot
\title{Ranks and eigenvalues of states with prescribed reduced states}

\author{Chi-Kwong Li${}^{1}$, Yiu-Tung Poon${}^{2}$, and Xuefeng Wang${}^{3}$ \\
${}^{1}$Department of Mathematics, College of William and Mary,  \\
Williamsburg, VA 23187, USA. (ckli@math.wm.edu)\\
${}^{2}$Department of Mathematics, Iowa State University, \\
Ames, IA 50011, USA. (ytpoon@iastate.edu) \\
${}^{3}$School of Mathematical Sciences,
Ocean University of China, \\
Qingdao,
Shandong 266100, China. (wangxuefeng@ouc.edu.cn)}

\date{}
\maketitle

\begin{abstract}
For a quantum state represented as an $n\times n$ density matrix $\sigma \in M_n$,
let $\cS(\sigma)$ be the compact convex
set of quantum states $\rho = (\rho_{ij})
\in M_{m\cdot n}$ with the first partial trace equal to $\sigma$, i.e.,
$\tr_1(\rho) = \rho_{11}  + \cdots + \rho_{mm} = \sigma$. It is known that
if $m \ge n$ then there is a rank one matrix $\rho \in \cS(\sigma)$
satisfying $\tr_1(\rho) = \sigma$.
If $m < n$, there may not be rank one matrix in $\cS(\sigma)$.
In this paper,
we determine the ranks of the elements and ranks of the extreme points of the set
$\cS$. We also determine $\rho^* \in \cS(\sigma)$ with rank bounded by $k$
such that $\|\tr_1(\rho^*) - \sigma\|$ is minimum for a given
unitary similarity invariant norm $\|\cdot\|$.
Furthermore, the relation between the eigenvalues of $\sigma$ and
those of $\rho \in \cS(\sigma)$ is analyzed.
Extension of our results and open problems will be mentioned.

\end{abstract}



\section{Introduction}

In quantum information science, quantum states are used to store,
process, and transmit information. Mathematically,  quantum states
are represented by density matrices, i.e., positive semidefinite
matrices of trace 1; see \cite{K,NC} for example.
Let $M_n$ ($H_n$) be the set of $n\times n$ complex (Hermitian)
matrices, and let
$D(n)$ be the set of density matrices in $M_n$.
Suppose $\sigma_1 \in D(m)$ and $\sigma_2 \in D(n)$ are two quantum states.
Their product state is  $\sigma_1 \otimes \sigma_2$.
The combined system is known as the bipartite system,
and a general quantum state is represented by a density matrix
$\rho \in D({m\cdot n})$.
Two basic quantum operations used to extract information of the subsystems
from a quantum state of the bipartite system are the partial traces, which are
linear maps satisfying
$$\tr_1(\sigma_1\otimes \sigma_2) = \sigma_2 \quad \hbox{ and } \quad
\tr_2(\sigma_1\otimes \sigma_2) = \sigma_1$$
on tensor states $\sigma_1 \otimes \sigma_2 \in D({m\cdot n})$.
Then for a general state
$\rho = (\rho_{ij})_{1 \le i, j \le m} \in D({m\cdot n})$
such that $\rho_{ij} \in M_n$, we have
$$\tr_1(\rho) = \rho_{11} + \cdots + \rho_{mm} \in M_n
\quad \hbox{ and } \quad
\tr_2(\rho) = (\tr \rho_{ij})_{1 \le i, j \le m} \in M_m.$$

It is well known that for every $\sigma \in D(n)$ there is a pure
state $\rho \in D({n\cdot n})$ such that $\tr_1(\rho) = \sigma$.
This is known as the {\it purification process},
which is useful in the study of quantum computation; for example see \cite{NC}.
In fact, it is easy to show that for every $\sigma \in D(n)$ of rank $r$, there is
a pure state $\rho \in D({r \cdot n})$ satisfying $\tr_1(\rho) = \sigma$.
However, one may not be able to find a purification if the dimension of the first
system is bounded, say, due to limitation of resource or restriction on
the physical system.  In such a case, two questions naturally arise:

\begin{problem} Can we find
a pure state $\rho \in D({m\cdot n})$ such that
$\tr_1(\rho)$ is nearest to $\sigma$,
say, with respect to a certain norm $\|\cdot\|$ on  $H_n$?
\end{problem}

\begin{problem} Can we find $\rho \in D({m\cdot n})$ with rank as low as possible
so that  $\tr_1(\rho) = \sigma$.
\end{problem}

In Section 2, we will give complete answers to these problems.
In particular,
for a given $\sigma \in D(n)$ and a given positive integers $k$ and $m$,
we determine
$$\min\{\|\tr_1(\rho) - \sigma\|: \rho \in D({m\cdot n}) \hbox{ has rank at most } k\}$$
for any {\it unitary similarity invariant norm} $\|\cdot\|$
on $H_n$, i.e., norm $\|\cdot\|$ such that $\|UA U^{*}\| = \|A\|$
for any $A \in H_{n}$ and  unitary $U \in M_n$.
In fact, using the notion of {\it majorization}, we obtain a general result
on the existence of $\rho \in D({m\cdot n})$ with low rank such that
$\sigma - \tr_1(\rho)$ satisfies many nice properties.

To better understand quantum states with a prescribed reduced state,
we consider the compact convex set
$$\cS(\sigma) = \{\rho \in D({m\cdot n}): \tr_1(\rho) = \sigma\}.$$
In Sections 3, we determine the ranks of elements and the ranks of extreme points
in $\cS(\sigma)$. In Section 4, we analyze the relationship between the eigenvalues of
$\sigma $ and those of the elements in $\cS(\sigma )$.
We obtain a necessary and sufficient condition relating the eigenvalues
of $\rho$ and $\sigma$ when $m \ge n$, and also in some low dimension cases.
The general problem for the case when $m < n$ remains open.
In Section 5, we discuss the extensions
and difficulties of the study to multi-partite systems.

Researchers have used advanced techniques in
representation theory (see \cite{Hy,Kl} and their references)
to give a complete description of the relationship between the
eigenvalues of the reduced states  $\tr_1(\rho), \tr_2(\rho)$,
and those of the ``parent'' state $\rho$. However, it is not easy to
generate (and store) all the inequalities even for a moderate size problem (see \cite{Kl}).
Moreover, it is not easy to use the numerous set of inequalities to answer basic
questions.
For example, for $(m,n) = (2,3)$, there is a density matrix
$\rho \in M_{2\cdot 3}$ and reduced states $\tr_2(\rho), \tr_1(\rho)$ with eigenvalues
$a_1 \ge \cdots \ge a_6$, $b_1\ge b_2$, and $c_1 \ge c_2 \ge c_3$ respectively
if and only if 41 inequalities are satisfied \cite{Kl}.
However, it is not easy to use the result to answer Problems 1 and 2,
and other simple problems such as:

\begin{enumerate}
\item Characterize the  eigenvalues $a_1 \ge \cdots \ge a_6$
of a density matrix $\rho \in M_{2\cdot 3}$
such that the (first) partial trace is a maximally entangled state, i.e.,
$\tr_1(\rho) = I_3/3$.

\item Determine all possible ranks of matrices in the convex set
$$\cS(I_3/3) = \{\rho \in D({2\cdot 3}): \tr_1(\rho) = I_3/3\}.$$

\item Determine the ranks 
of the extreme points of the convex set $\cS$ above.

\end{enumerate}

Nevertheless, one can readily answer the above problems using
our results in Sections 3 and 4. (See Section 5.)

We conclude this section by fixing some notations.
We will use $X^t$ and $X^*$ to denote the
transpose and conjugate transpose of a matrix or vector $X$.

Let
$\{e_1^{(m)}, \dots, e_m^{(m)}\}$ and
$\{e_1^{(n)}, \dots, e_n^{(n)}\}$ be the standard bases
for $\IC^m$ and $\IC^n$, respectively. Then, clearly,
$\{e_1^{(m)}\otimes e_1^{(n)},
e_1^{(m)}\otimes e_2^{(n)}, \dots, e_m^{(m)}
\otimes e_n^{(n)}\}$ is the standard basis for
$\IC^{m}\otimes \IC^n \equiv \IC^{mn}$.
For $\ell =m,n$ and $1\le i,\ j\le \ell$, let $E^{(\ell)}_{i\,j}=e^{(\ell)}_i(e^{(\ell)}_j)^t$.
Then $\{E^{(\ell)}_{i\,j}:1\le i,\ j\le \ell\}$  is the standard basis for $M_{\ell}$.
For simplicity, we use the notation $e_i$ for $e_i^{(m)}$ or $e_i^{(n)}$ and $E_{i\, j}$ for
$E^{(\ell)}_{i\,j}$,
if the dimension   is clear in the context.
Also, we use $e_i\otimes e_j$
instead of $e_i^{(m)}\otimes e_j^{(n)}$.

Furthermore, we use $PSD(n)$ and $\cR_k(n)$ to denote the sets of matrices in $M_n$
which are positive semidefinite  and have rank at most  $k$, respectively.

Two linear maps
$$[\,\cdot\,]:
\IC^{mn} \rightarrow M_{n,m} \quad \hbox{ and } \quad \vec: M_{n,m} \rightarrow \IC^{mn}$$
will be used frequently in our discussion. Here,
for $w = (w_1, \dots, w_{mn})^t \in \IC^{mn}$
$W = [w]$ is the  $n\times m$ matrix such that the $j$th column equals
$(w_{(j-1)n+1}, \dots, w_{jn})^t$ for $j = 1, \dots, m$;
and $\vec$ is the inverse map  which converts
an $n\times m$ matrix $W$ to $w = \vec(W) \in \IC^{mn}$
so that $W = [w]$.
Note that
$$\tr_1(ww^*) = W W ^* \quad \hbox{ and } \quad
\tr_2(ww^*) = W^t(W^t)^*.$$

\section{Approximation by reduced states of low rank states}

To state and prove our results, we need the following definitions and notation.

Recall that for $x, y \in \IR^n$, $x$ is {\it majorized} by $y$, denoted by
$x \prec y$, if the sum of entries of the vectors are the same,
and the sum of the $k$ largest entries of $x$ is not
larger than that of $y$ for $k = 1, \dots, n-1$.
A scalar function $f: \IR^n \rightarrow \IR$ is {\it Schur convex} provided
$f(x) \le f(y)$ whenever $x \prec y$.

We can extend the definition
of majorization and Schur convex function to Hermitian matrix as follows.
For every $A \in H_n$, let $\lambda(A)\in \IR^n$ be the vector of eigenvalues
of $A$ with entries arranged in descending order.
For $A, B \in H_n$, we write $A \prec B$ if $\lambda(A) \prec \lambda(B)$.
A function $f: \IR^n \rightarrow \IR$ can be extended to
$\tilde f: H_n \rightarrow \IR$ by setting $\tilde f(A) = f(\lambda(A))$.
On the other hand, some scalar functions on $H_n$ or $D(n)$ can be viewed as
an extension of $f: \IR^n \rightarrow \IR$.
For example, the determinant function $A \mapsto \det(A)$ on $H_n$
corresponds to  $f(x_1, \dots, x_n) = \prod_{j=1}^n x_j$;
the von Neumann entropy $\rho \mapsto -\tr \rho( \log \rho)$ on $D(n)$
corresponds to $f(x) = -\sum_{j=1}^n x_j \log x_j$ with the convention that
$x_j \log x_j = 0$ if $x_j = 0$. Moreover, every unitary similarity invariant norm
$\|\cdot\|$ corresponds to a Schur convex norm function $f: \IR^n \rightarrow \IR$;
see \cite{LT}.
For example, for $1 \le p \le \infty$ the Schatten $p$-norm defined by
$$\|A\|_p = \{\tr |A|^p\}^{1/p}$$
is a unitary similarity invariant norm,
where $|A|$ is the unique positive semi-definite matrix such that
$|A|^2 = A^*A$. Here, we take the limit $p \rightarrow \infty$,
and set $\|A\|_\infty = \max\{ |\mu|:
\mu \hbox{ is an eigenvalue of } |A|\}$.
Clearly, the Schatten $p$-norm corresponds to the
$\ell_p$ norm on $\IR^n$ defined by $\ell_p(x_1, \dots, x_n) = (\sum_{j=1}^n |x_j|^p)^{1/p}$.

We have the following result.

\begin{theorem} \label{pro:appro1} Let $n,m,k$ be positive integers such that
$k \le m$. Suppose $\sigma \in D(n)$ has rank $r$ and has spectral decomposition
$\sum_{j=1}^r \lambda_j x_j x_j^*$ with $\lambda_1 \ge \cdots \ge \lambda_r > 0$.
Then there is
$\rho \in D({m\cdot n})$ with rank at most $k$ such that
$\tr_1(\rho) = \sigma$ if and only if $r \le mk$.

If $mk < r$, then there is
$\rho \in D({m\cdot n})$ with rank $k$ such that
$$\tr_1(\rho) = \sum_{j=1}^{mk} (\lambda_j + \mu) x_j x_j^*,$$
where $\mu=(\sum_{j=mk+1}^{r}  \lambda_j) /(mk)$,
so that

\begin{equation}\label{maj1}\lambda(\sigma-\tr_1(\rho)) =(\lambda_{mk+1}, \dots, \lambda_{r},
\underbrace{0,\dots,0}_{\tiny n-r\mbox{\tiny \  terms}},
\underbrace{-\mu, \dots, -\mu}_{\tiny mk\mbox{\tiny \ terms}} )
\prec \lambda(\sigma-\tr_1(\tilde \rho)  )\end{equation}
for all
$\tilde\rho \in D({m\cdot n})$ with rank at most $k$.
\end{theorem}

By the properties of Schur convex functions and unitary similarity
invariant norm (see \cite{MO} and \cite{LT}), we immediately have the following.

\begin{corollary}
Suppose $\sigma$ and $\rho$ satisfy the hypothesis and conclusion on
Theorem {\rm \ref{pro:appro1}}.
Then for every Schur convex function $f: \IR^n \rightarrow \IR$, we have
$$f(\lambda(\sigma-\tr_1(\rho))) \le f(\lambda(\sigma-\tr_1(\tilde \rho)))
\quad \hbox{ for all } \tilde \rho \in D({m\cdot n}) \hbox{ of rank at most } k.$$
Furthermore, for every  unitary similarity invariant norm $\|\cdot\|$ on $H_n$, we have
$$\|\sigma-\tr_1(\rho)\| \le \|\sigma-\tr_1(\tilde \rho)\|
\quad \hbox{ for all } \tilde \rho \in D({m\cdot n}) \hbox{ of rank at most } k.$$
\end{corollary}

\it Proof of Theorem \ref{pro:appro1}. \rm
If $r \le mk$, then we can write $\sigma = \sigma_1 + \cdots + \sigma_k$, where each
$\sigma_i$ has rank at most $m$ and has a purification $\rho_i \in D({m\cdot n})$.
Then $\rho = \rho_1 + \cdots + \rho_k \in D({m\cdot n})$ has rank at most $k$ such that
$\tr_1(\rho) = \sigma$.

Conversely, if $\rho \in D({m\cdot n})$ has rank at most $k$
so that it is the sum of at most $k$ rank one matrices $\rho_1, \dots, \rho_k$.
Then $\tr_1(\rho_i)$ has rank at most $m$, and $\tr_1(\rho)$ has rank at most $mk$.

Suppose $mk < r$. Let $\hat \sigma = \sum_{j=1}^{mk} (\lambda_j + \mu) x_j x_j^* \in H_n$.
Then $\hat \sigma = \hat \rho_1 + \cdots + \hat \rho_k$ such that each
$\hat \rho_j$ has rank $m$, and admits a purification $\rho_j \in D({m\cdot n})$.
Let $\rho = \rho_1 + \cdots + \rho_k$.
Then $\rho$ has rank at most $k$ and $\tr_1(\rho) = \hat \sigma$.

To prove (\ref{maj1}), suppose $r > mk$.  Let
$$(c_1, c_2\,\dots, c_{n})=\lambda(\sigma-\tr_1(\rho))=(\lambda_{mk+1}, \dots, \lambda_{r},
\underbrace{0,\dots,0}_{\tiny n-r\mbox{\tiny \  terms}},
\underbrace{-\mu, \dots, -\mu}_{\tiny mk\mbox{\tiny \ terms}} ),$$
where $\mu=(\sum_{j=mk+1}^{r}  \lambda_j) /(mr)$.

Suppose  $\tilde \rho$ has rank at most $k$. Then $\tr_1(\tilde \rho)$ has rank at most $mk$. Let
$$\lambda(\tr_1(\tilde \rho))=(b_1,\dots,b_{n}).$$
Then we have $b_i=  0$ for $mk <i \le n$.

Suppose $\lambda(\sigma - \tr_1(\tilde \rho)) = (a_1, \dots, a_n).$  We will prove that
\begin{equation}\label{maj2}(c_1, c_2\,\dots, c_{n})\prec (a_1, a_2\,\dots, a_{n}).
\end{equation}
Clearly, we have $\sum_{i=1}^{mn}c_i=0=\sum_{i=1}^{mn}a_i$.
Since $\sigma=(\sigma - \tr_1(\tilde \rho))+\tr_1(\tilde \rho)$, by Wielandt's inequalities \cite[Theorem 9.G.1a]{MO},
for $1\le s\le n-mk$, we have
$$\sum_{i=1}^sa_i
=\sum_{i=1}^sa_i+\sum_{i=1}^s b_{mk+i}
\ge\sum_{i= 1}^{s}\lambda_{mk+i}
=\sum_{i=1}^sc_i.$$
Let $\tilde \mu =
(\sum_{j=1}^{n-mk}  a_j) /(mk)=- (\sum_{j=n-mk+1}^{n} a_j)/(mk)$.
Then we have
$$(c_1, c_2\,\dots, c_{n})\prec (a_1, a_2\,\dots, a_{n-mk},
\underbrace{-\tilde\mu, \dots, -\tilde\mu}_{\tiny mk\mbox{\tiny \ terms}} )
\prec (a_1, a_2\,\dots, a_{n}).$$
\vskip -.5in \qed

\section{Ranks of elements in $\cS(\sigma)$}

In this section, for $\sigma\in D(n)$, we consider the compact convex set
$$\cS(\sigma) = \{\rho\in D(m\cdot n): \tr_1(\rho)= \sigma\}.$$
We will completely determine the
ranks attainable by its elements and by its extreme points.
The following lemma is useful in our discussion.

\begin{lemma}\label{lem:Connect-Extension-UST}
Let $\sigma \in D(n)$ and $U \in M_n$ be unitary.
Then
$$\cS(U\sigma U^*) = (I_m\otimes U) \cS(\sigma) (I_m\otimes U)^*
= \{ (I_m\otimes U)\rho (I_m\otimes U)^*: \rho \in \cS(\sigma)\}.$$
\end{lemma}

Recall that $\sigma\in D(n)$ is a pure state if  $\rank(\sigma)=1.$
It is well known that the extreme points of $D(n)$ are
pure states. For a pure state $\sigma \in D(n)$, we have the
following complete description of $\cS(\sigma)$. In particular,
all states in the set $\cS(\sigma)$ are tensor states.

\begin{proposition}\label{pure}
Let $\sigma \in D(n)$ be a pure state. Then
$$\cS(\sigma) = \{\xi\otimes\sigma: \xi \in D(m)\}.$$
Consequently, there is $\rho \in \cS(\sigma)$ with rank $k$ if and only if
$1 \le k \le m$. Moreover, $\rho$ is an extreme point of $\cS(\sigma)$
if and only if $\rho = \xi \otimes \sigma$ for a rank one $\xi \in D(m)$.
\end{proposition}

\begin{proof}
By Lemma \ref{lem:Connect-Extension-UST}, we may assume that
$\sigma =   E_{11}\in M_n$.
Then $\rho = (\rho_{ij})_{1 \le i, j \le m} \in D(m\cdot n)$ with $\rho_{ij} \in M_n$
if and only if $\rho_{11}+\rho_{22}+\cdots+\rho_{mm}=E_{11}$. Since $\rho$ is positive semidefinite,
 we see that $\rho_{ii} =\xi_iE_{11}$, where $\xi_i\ge 0$ for $i = 1,\dots,\ m$ and $\sum_{i=1}^m\xi_i=1$. Thus
$\rho_{ij} = \xi_{ij}E_{11}$ for some $\xi_{ij}$, $i,\ j=1,\dots,\ m$, with $\xi_{ii}=\xi_i$. Hence, $\rho = \xi \otimes\sigma$
with $\xi = \(\xi_{ij}\)=\tr_2 (\rho)\in D(m)$.

Clearly, $\rank(\rho) = \rank(\xi) \in \{1,\dots,\ m\}$.  Also, it is well known
that $D(m)$ is a compact convex set with the pure states as the set of extreme points.
The last statement follows.
\end{proof}

For a general state
$\sigma \in D(n)$, it is not so easy to give a complete description for
the set $\cS(\sigma)$.
In the following, we consider general states $\sigma \in D(n)$
and determine the ranks and extreme points
of matrices in  $\cS(\sigma)$.

\begin{theorem} \label{pro:rank1}
Let $\sigma \in D(n)$ have rank $r$.
There is $\rho \in \cS(\sigma) \subseteq D(m\cdot n)$
with rank $k$  if and only if
$$\lceil r/m \rceil \le k \le  rm.$$
In particular, if there are matrices in $\cS(\sigma)$ of rank $r_1, r_2$ with $r_1 < r_2$,
then there are matrices in $\cS(\sigma)$ of rank $r_1+1, \dots, r_2-1$.
\end{theorem}

\begin {proof}
By Lemma \ref{lem:Connect-Extension-UST}, we may
assume $\sigma = \diag(d_1, \dots, d_r, 0, \dots, 0)$
with $d_1 \ge \cdots \ge d_r > 0$.

\marginpar{\tiny }

Suppose $\rho = ZZ^*\in \cS(\sigma)$ such that $Z$ is $mn\times k$, where $k$
is the rank of $\rho$ and $Z$ has columns $z_1, \dots, z_k \in \IC^{mn}$.
Set $Z_j = [z_j]$ for $j = 1, \dots, k$.
Then $\sigma =\tr_1(\rho) = \sum_{j=1}^k Z_j Z_j^*$ has rank at most
$mk$ because every $Z_jZ_j^*$ has rank at most $m$.
Hence,  $r/m \le k$.

Next, we consider the upper bound for $k$.
Suppose $\rho = (\rho_{ij}) \in \cS(\sigma)$
with $\rho_{ij} \in M_n$. Since
$$\sigma = \diag(d_1, \dots, d_r, 0, \dots, 0)=\rho_{11}+\rho_{22}+\cdots+\rho_{mm}\,,$$
  we have
$ \rho_{ij}  \in $ span $\{E_{pq}:1\le p,\ q\le r\}$ for all $1\le i,\ j\le m$. Hence,
$\rho = (\rho_{ij})$ has rank at most $rm$.

Finally, we show that for every $k$ between the lower and  upper bound,
there exists $\rho \in \cS(\sigma)$ with rank $k$.
Suppose $r/m \le k \le rm$. Then $\rho $ can be  constructed as follows.

\rm
{\bf Case 1.} Suppose $r<k\leq rm$ and denote $k = qr + s$ with  $0 < q < m$
and $0 < s \le  r$.
Let $$\rho=\sum_{i=1}^{q+1}\sum_{j=1}^s\ \dfrac{d_j}{q+1}E_{ii}^{(m)}\otimes E_{jj}^{(n)}+
\sum_{i=1}^q\sum_{j=s+1}^r\ \dfrac{d_j}{q}E_{ii}^{(m)}\otimes E_{jj}^{(n)}\,.$$
Then rank $\rho=(q+1)s+q(r-s)=qr+s=k$ and
$$\tr_1(\rho)=\sum_{i=1}^{q+1}\sum_{j=1}^s\ \frac{d_j}{q+1}  E_{jj}^{(n)}+
\sum_{i=1}^q\sum_{j=s+1}^r\ \frac{d_j}{q}  E_{jj}^{(n)}
= \sum_{j=1}^r\ d_jE_{jj}^{(n)}=\sigma\,.$$

{\bf Case 2.} Suppose $r/m \le k \le r \le n$,
and $r = k \hat q + \hat s$ with
$0 \le \hat q < m$ and $1 \le \hat s \le k$.
Let $f_{j}=\sqrt{d_{ j}}e^{(n)}_j$ for $1\le j\le n$, and

$$
\rho=\sum_{j=1}^s
\(\sum_{i=1}^{\hat q+1}\ e^{(m)}_i\otimes f_{(i-1)k+j}\)
\(\sum_{i=1}^{\hat q+1}\ e^{(m)}_i\otimes f_{(i-1)k+j}\)^*+
$$
$$\hskip 1in
\sum_{j=\hat s+1}^{k}\(\sum_{i=1}^{\hat q}\ e^{(m)}_i\otimes f_{(i-1)k+j}\)
\(\sum_{i=1}^{\hat q}\ e^{(m)}_i\otimes f_{(i-1)k+j}\)^*
$$
Then rank $\rho=\hat  s+(k-\hat s)= k$ and
$$\tr_1(\rho)= \sum_{j=1}^{\hat s}\sum_{i=1}^{\hat q+1}\ d_{(i-1)k+j}E^{(n)} _{(i-1)k+j
\, \ (i-1)k+j}+
\sum_{j=\hat s+1}^{k}\sum_{i=1}^{\hat q}\
d_{(i-1)k+j}E^{(n)} _{(i-1)k+j \, \  (i-1)k+j}
=\sum_{\ell=1}^{r} d_{\ell}E^{(n)} _{\ell\,\ell}=\sigma\,.$$
{\hfill ~} \end{proof}

By Theorem \ref{pro:rank1}, we have the following corollary, which is part of Theorem
\ref{pro:appro1}.

\begin{corollary} \label{cor:2.4}
Suppose $\sigma \in D(n)$ has rank $r$. Then there is $\rho \in D({m\cdot n})$
with rank not larger than $k$ such that $\tr_1(\rho) = \sigma$ if
and only if $km \ge r$. In particular, there  $\sigma$ has a purification
$\rho \in D(m\cdot n)$ if and only if $m \ge r$.
\end{corollary}

Next, we consider the extreme points of the set $\cS(\sigma)$.
We begin with some general observations.

\begin{lemma} \label{lem1}
Let $\sigma\in D(n)$ and let $\rho \in
\cS(\sigma) \subseteq D(m\cdot n)$. Then
$\rho$ is not an extreme point if and only if
there exists a nonzero $\xi\in H_{mn}$ such that
$\rho \pm \xi \in PSD(m\cdot n)$ and $\tr_1(\xi) = O_n$.
In such a case, there are $\rho_1, \rho_2 \in \cS(\sigma)$
with $\rank(\rho_1) < \rank(\rho)$ such that
$\rho = (\rho_1+\rho_2)/2$.
\end{lemma}

\it Proof. \rm If $\rho \in \cS(\sigma)$ is not extreme, then
there are two different elements $\rho_1, \rho_2 \in \cS(\sigma)$ such that
$\rho = (\rho_1+\rho_2)/2$. Let $\xi = (\rho_1-\rho_2)/2\ne 0$.
Then $\rho\pm \xi \in \cS(\sigma)$ so that
$\rho\pm \xi \in PSD(m\cdot n)$ and
$\tr_1(\xi) = \tr_1(\rho_1-\rho_2)/2 =(\sigma-\sigma)/2 = O_n$.
Conversely, if $\xi \in H_{mn}$ satisfies
$\rho \pm \xi \in D(m\cdot n)$ and $\tr_1(\xi) = O_n$, then
we can set $\rho_{\pm} = \rho \pm \xi$ so that
$\rho_+, \rho_- \in \cS(\sigma)$ and
$\rho = (\rho_++\rho_-)/2$.

Now, if $\rho\in \cS(\sigma)$ has rank $r$ and is not an extreme point.
Then we can choose an orthonormal set $\{z_1,\dots,z_r\}$ in $\IC^{mn}$ such that
$\rho= \sum_{j=1}^r\lambda_j z_j z_j^*$.
Suppose a nonzero $\xi \in H_{mn}$
is such that
$\rho\pm \xi \in PSD(m\cdot n)$ and $\tr_1(\xi) = O_m$.
Then $\xi = \sum_{1 \le i, j \le r} h_{ij} z_i z_j^*$ for some non-zero $\(h_{ij}\) \in H_{r}$.
Thus there exists $t>0$ such that
\begin{itemize}
\item[{\rm 1)}]
$\rho\pm t\xi \in PSD(m\cdot n) $,
\item[{\rm 2)}] either $\rho_1=\rho+ t\xi$ or $\rho_2=\rho- t\xi$ has rank $<r$, and
\item[{\rm 3)}] $\rho = (\rho_1+\rho_2)/2$, with $\rho_1, \rho_2 \in \cS(\sigma)$.
\end{itemize}
The last assertion follows.
\qed

\begin{theorem} \label{extreme} Suppose $\rho \in \cS(\sigma)$ for a given
$\sigma \in D(n)$ such that $\rho$ has rank $r$ and $\rho = ZZ^* \in D(m\cdot n)$, where
$Z$ has columns $z_1, \dots, z_r \in \IC^{mn}$.
Then $\rho$ is an extreme points of $\cS(\sigma)$ if and only if
the set $T(z_1,\dots, z_r) = \left\{[z_i][z_j]^*: 1 \le i, j \le r\right\}$ is linearly independent.
\end{theorem}

\it Proof. \rm Suppose $T(z_1, \dots, z_r)$ is linearly dependent.
Then there is $H = (h_{ij}) \in M_r$ such that
$$\sum_{i,j} h_{ij}[z_i][z_j]^*  = 0.$$
Let $[z_j] = Z_j$ for $j = 1, \dots, r$.
Then $[Z_1 \cdots Z_r] (H \otimes I_m)[Z_1 \cdots Z_r]^* = 0$.
We may replace $H$ by $e^{it}H + e^{-it}H^*$ for a suitable
$t \in [0, 2\pi)$ and assume that $0 \ne H = H^*$.
Then for $t > 0$ such that $\|t H\| < 1$,
$\rho_{\pm} = \rho \pm tZHZ^* = Z(I_r\pm tH)Z^*$ is positive semidefinite.
Moreover,
$$\tr_1(\rho_{\pm}) = \sum_{j=1}^r Z_jZ_j^* \pm
t\sum_{1 \le i, j \le r} h_{ij} Z_iZ_j^* = \sigma$$
and $\tr(\rho_{\pm}) = \tr(\sigma) = 1$. Thus, $\rho_{\pm} \in \cS(\sigma)$
are two different elements such that $\rho = (\rho_{+} + \rho_{-})/2$.
Hence, $\rho$ is not an extreme points.

Conversely, if $\rho$ is not an extreme point of $\cS(\sigma)$, then
$\rho = (\rho_+ + \rho_-)/2$ for two different elements $\rho_+ ,\ \rho_-$ in $\cS(\sigma)$.
Then $\rho_+ -\rho = \rho - \rho_- = \tilde H\ne 0$ so that
$\rho_+ = \rho + \tilde H = ZZ^* + \tilde H \in \cS(\sigma)$
and $\rho_- = \rho - \tilde H = ZZ^* - \tilde H \in \cS(\sigma)$.
Thus, the range space of $\tilde H$ is a subspace of the range space of
$\rho$, which is the column space of $Z$. Thus,
$\tilde H$ has the form $ZHZ^*$ for some $0 \ne (h_{ij}) =
H = H^*\in M_r$ so that $I_r \pm H$
are positive semidefinite. Moreover,
$$\tr_1(\rho_{\pm}) = \tr_1(\rho \pm \tilde H) = \tr_1(\rho)= \sigma.$$
It follows that $0 = \tr_1(\tilde H) = \sum_{ij} h_{ij} Z_i Z_j^*$.
Hence, $T(z_1, \dots, z_r)$ is linearly dependent. \qed

Next, we determine all possible ranks of the extreme points of $\cS(\sigma)$.

\begin{theorem}\label{the:Partial-extreme-point}
Suppose $\sigma\in D(n)$ and $\rank(\sigma)=r.$
There is an extreme point $\rho \in \cS(\sigma) \subseteq D(m\cdot n)$
with rank $k$ if and only if
$$\lceil r/m \rceil \le k \le r.$$
Moreover, every $\rho \in \cS(\sigma)$ with rank equal to $\lceil r/m \rceil$
is an extreme point.
For $\lceil r/m \rceil < k  \le  r$, there exists $\rho \in \cS(\sigma)$ which is
not an extreme point.
\end{theorem}

\begin{proof}
By Lemma \ref{lem:Connect-Extension-UST}, we may assume that  $\sigma = \diag(d_1, \dots, d_r, 0, \dots, 0)$
with $d_1 \ge \cdots \ge d_r > 0$.

\medskip\noindent
(1)  We show that any $\rho\in\cS(\sigma)$ with $\rank(\rho)=k > r$ is not
an extreme point.

Suppose $\rho=z_1z_1^*+\cdots+z_kz_k^*$. Let $Z_i = [z_i]$
for $i = 1, \dots, k$. Then $\sum_{j=1}^k Z_jZ_j^* = \sigma$.
It follows that the last $n-r$ rows of $Z_i$ are zero for $i = 1, \dots, k$.
Thus, $Z_iZ_j^* = C_{ij} \oplus O_{n-r}$ for some
$C_{ij} \in M_r$. Thus, $\{Z_iZ_j^*: 1 \le i, j \le k\}$
is linearly dependent as $k^2 > r^2$.
By Theorem \ref{extreme}, $\rho$ is not an extreme point.

\medskip
\noindent
(2) Suppose $r/m \le k \le r$.  We show that there is
an extreme point $\rho\in\cS(\sigma)$ with  $\rank(\rho)=k$.

Because $r/m \le k \le r$, we can let  $r=k\hat q+\hat s,$
and use the construction in Case 2 in the proof of Theorem
\ref{pro:rank1}
to obtain $\rho = \sum_{j=1}^k  z_jz_j^*$.
Note that for $1 \le i, j \le k$,
$[z_i][z_j]^*$ has the form
$\sqrt{\lambda_i\lambda_j} E_{ij}{(k)} \oplus Y_{ij}$.
Thus, $\{[z_i][z_j]^*: 1 \le i, j \le k\}$
is linearly independent, and $\rho$ is an extreme point.

\medskip
\noindent
(3) Suppose $\lceil r/m \rceil < k \le r$.  We show that there is
 $\rho\in\cS(\sigma)$ with  $\rank(\rho)=k$ such that
 $\rho$ is not an extreme point.

Because $\lceil r/m\rceil < k \le r \le n$,
we may use the the construction
in  Case 2 in the proof of Theorem
\ref{pro:rank1}, with $k$ replaced by $k-1$ to get
$\tilde \rho = \sum_{j=1}^{k-1} z_j z_j^*$ such that
$\tr_1(\tilde \rho) = \sigma$. Since $k-1<r$,
  $Z_1 = [z_1]$ has two nonzero columns.
Replace $z_1$ by $\tilde z_1 = z_1/\sqrt{2}$
and construct $z_k$ so that $Z_k = [z_k]$ is obtained from
$[\tilde z_1]$ by multiplying its first column by $-1$.
Then $\rho = \tilde z_1 \tilde z_1^* + \sum_{j=2}^k z_j z_j^* \in
\cS(\sigma)$ has rank $k$. Note that $[\tilde z_1][\tilde z_1]^* =
[z_k][z_k]^*$ so that
$T(\tilde z_1, z_2, \dots, z_k)$ is linearly dependent.
So, $\rho$ is not extreme.

\medskip\noindent
(4) We show that if $\rho\in\cS(\sigma)$ has rank
$k=\lceil r/m \rceil$, then $\rho$ is an extreme point.

If $\rho$ is not an extreme point, then by Lemma \ref{lem1}
$\rho_1, \rho_2 \in \cS(\sigma)$ with $\rank(\rho_1) < \rank(\rho)$
such that $\rho = (\rho_1+\rho_2)/2$,
which is a contradiction.
\end{proof}

\begin{corollary}
Suppose $\sigma \in D(n)$ and $\rho \in \cS(\sigma) \subseteq D(m\cdot n)$.
\begin{enumerate}
\item[{\rm (a)}] If $\rho$ has rank one, then $\rho$ is an extreme point of $\cS(\sigma)$.
\item[{\rm (b)}] If $\rho$ has rank $k > n$, then $\rho$ is not an extreme point.
\end{enumerate}
\end{corollary}

\it Proof. \rm (a) If $\rho = zz^*$, then $\{[z][z]^*\}$ is linearly independent.
So, $\rho$ is an extreme point.

(b) If $\rho = ZZ^*$, where $Z$ has linearly independent columns $z_1, \dots, z_k$,
then $T(z_1, \dots, z_k) \subseteq M_{n^2}$
has $k^2$ elements with $k^2 > n^2$, and hence is a linearly dependent set in
$M_{n^2}$. So, $\rho$ is not an extreme point. \qed

\section{Eigenvalues}

As mentioned in the introduction, even though we know the inequalities governing
the eigenvalues of $\rho \in D({m\cdot n})$ and those of
$\sigma_2 = \tr_1(\rho)$ and $\sigma_1 = \tr_2(\rho)$, it is not easy to
use them to determine the relations between the eigenvalues of $\rho$
and $\tr_1(\rho)$ (without specifying those of $\tr_2(\rho)$).
We have the following result.

\begin{theorem} \label{mn1} Suppose $m, n \ge 2$,
$\lambda_1 \ge \cdots \ge \lambda_n \ge 0$
and $\mu_1 \ge \cdots \ge \mu_{mn} \ge 0$ satisfy $\sum_{j=1}^n \lambda_j = 1
= \sum_{j=1}^{mn} \mu_j$.
\begin{itemize}
\item[{\rm (a)}]
If there exist $\sigma \in D(n)$ with eigenvalues $\lambda_1, \dots, \lambda_n$
and  $\rho \in S(\sigma) \subseteq D(m\cdot n)$ with eigenvalues
$\mu_1, \dots, \mu_{mn}$, then
\begin{equation}\label{major-1}
\left(\frac{\lambda_1}{m},\cdots,\frac{\lambda_1}{m},\frac{\lambda_2}{m},\cdots,
\frac{\lambda_2}{m},\cdots,\frac{\lambda_n}{m},\right)\prec
\left( \mu_1,\mu_{2}, \dots,   \mu_{ n m  }\right),
\end{equation}
so that
\begin{equation}\label{major}
(\lambda_1, \dots, \lambda_n) \prec \left(\sum_{j=1}^m \mu_j,
\sum_{j=1}^m \mu_{m+j}, \dots, \sum_{j=1}^m \mu_{(n-1)m+j }\right),
\end{equation}
and setting $\lambda_j = 0$ for $j > n$, we have
\begin{equation}\label{major2}
(\mu_1, \dots, \mu_{mn}) \prec
\left(\sum_{j=1}^m \lambda_j,
\sum_{j=1}^m \lambda_{m+j}, \dots, \sum_{j=1}^m \lambda_{m^2n-m+j }\right).
\end{equation}

\item[{\rm (b)}] If $m \ge n$ and
condition {\rm (\ref{major})} holds, then there exist $\sigma \in D(n)$ with
eigenvalues $\lambda_1, \dots, \lambda_n$ and $\rho \in \cS(\sigma)$ with
eigenvalues $\mu_1, \dots, \mu_{mn}$.
\end{itemize}
\end{theorem}

\it Proof. \rm
(a)  Suppose $\rho = (\rho_{ij})_{1 \le i, j \le m}\in S(\sigma) $ has eigenvalues
$\mu_1 \ge \cdots \ge \mu_{mn}$.  We may assume that
$\rho_{11 } +\rho_{22 }+\cdots+\rho_{ mm}  = \diag(\lambda_1, \dots, \lambda_n)$.
Then there is a permutation matrix $P \in M_{mn}$ such that
$P\rho P^t = (\tilde \rho_{ij})_{1 \le i, j \le n}$ such that
$\tilde \rho_{ij} \in M_m$ such that $\tr \tilde \rho_{jj} = \lambda_j$ for
$j = 1, \dots, n$. There are unitary matrices $U_1, \dots, U_n \in M_m$
such that all the diagonal entries of
$U_j \tilde \rho_{jj} U_j^*$ equals $\tr(\rho_{jj})/m = \lambda_j/m$.
Let $U = U_1 \oplus \cdots \oplus U_n$. Then
the vector of diagonal entries of the matrix $UP\rho P^t U^*$ is majorized by
the vector of eigenvalues; see \cite{H} and \cite[Chapter 5]{MO} for example.
We get (\ref{major-1}), and (\ref{major}).

To prove (\ref{major2}), suppose $\rho$ has spectral decomposition
$\rho=\mu_1z_1z_1^*+\cdots+\mu_{mn}z_{mn}z_{mn}^*$.
Then
$$\sum_{j=1}^k \mu_j = \tr(\sum_{j=1}^k \mu_j z_j z_j^*)
= \tr(\tr_1(\sum_{j=1}^k \mu_j z_j z_j^*)).$$
Because
$\tr_1(\sum_{j=1}^k \mu_j z_j z_j^*)$ has rank at most $mk$ and
$\tr_1(\rho) - \tr_1(\sum_{j=1}^k \mu_j z_j z_j^*)$ is positive semi-definite,
$\tr(\tr_1(\sum_{j=1}^k \mu_j z_j z_j^*))$ is bounded by the sum of
the $mk$ largest eigenvalues of $\tr_1(\rho)$, i.e.,
$\sum_{j=1}^{km} \lambda_j$.

(b) Suppose $m \ge n$ and the majorization holds. Let
$w_k = \sum_{j=1}^m \mu_{(k-1)m + j}$ for $k = 1, \dots, n$.
By the result of Horn \cite{H}, there exist  unitary matrices
$U_1,\dots,U_n \in M_m$ such that
$$A_k=U_k^*\diag(\mu_{(k-1)m + 1}, \mu_{(k-1)m + 2}\dots, \mu_{km})U_k$$
has constant diagonal $\dfrac{1}m(w_k,\dots,w_k)$.
 Then the matrix
$A=\sum_{k=1}^nA_k\otimes E^{(n)}_{kk}$ has eigenvalues $\mu_1,\dots,\mu_{mn}$
and has the form $A=(A_{ij})_{i,j=1}^m$, where
$A_{ii}=\dfrac{1}m\diag(w_1, \dots, w_n)\in M_n$. Let $U$ be a unitary such that
$U^*\diag(w_1, \dots, w_n)U$ has diagonal entries
$\lambda_1, \dots, \lambda_n$.
Let $\omega=e^{\frac{2\pi i}m}$ and $D = \oplus_{k=0}^{m-1} \diag(\omega^k,\omega^{2k},\dots, \omega^{nk})$. Then
$$\rho=D^*(I_m\otimes U)^*A(I_m\otimes U)D$$
will  have reduced state $\tr_1(\rho)
= \diag(\lambda_1, \dots, \lambda_n)$.
\qed

The following corollaries are clear.

\begin{corollary} Suppose $m = n = 2$, and $\sigma \in D(2)$ has eigenvalues
$\lambda_1 \ge \lambda_2 \ge 0$.  There exists $\rho \in D(2\cdot 2)$
with eigenvalues $\mu_1 \ge \cdots \ge \mu_4$ satisfying $\tr_1(\rho) = \sigma$
if and only if $\mu_1 + \mu_2 \ge \lambda_1$
\end{corollary}

\begin{corollary} Suppose $m \ge n$.
\begin{itemize}
\item[{\rm (a)}]
For any $\sigma \in D(n)$ there is a pure state $\rho \in D(m\cdot n)$
such that $\tr_1(\rho) = \sigma$.
\item[{\rm (b)}] If $\sigma \in D(n)$ is a pure state and
$\rho \in \cS(\sigma)$, then $\rho$ has rank at most $m$.
\end{itemize}
\end{corollary}

It is interesting to note that if $m \ge n$,
the simple majorization condition (\ref{major2})
governs the relations between the eigenvalues of $\rho$ and $\sigma$ with
$\rho \in \cS(\sigma)$.
For $m < n$, the majorization condition is not good enough as shown in the following.

\begin{example}
\rm Suppose $m =2$ and $ n = 3$.  Let $\sigma = I_3/3$, and $\rho = uu^*$ for a unit vector.
Then $\tr_1(\rho)$ has rank at most two and cannot be $\sigma$.
Note that the rank is not the only obstacle. Suppose
$\rho = U^*\diag(1-5d, d,d,d,d,d)U = (\rho_{ij})_{1 \le i,j \le 2}$
for $d = 0.1$, and $\tr_1(\rho) = \sigma$.
Since $\rho_{11} + \rho_{22} = I_3/3$, they commute and we may assume that
they are in diagonal form: $\rho_{11} = \diag(d_1,d_2,d_3)$,
and $\rho_{22} = I_3/3 - \rho_{11}$.
If $\rho_{11}$ has eigenvalues $d_1 \ge d_2 \ge d_3$, then
by the generalized interlacing inequality \cite{FP},
$d \ge d_2 \ge d\Ra d_2=d$. Similarly, the second largest eigenvalue of $\rho_{22}$
also equals $d$. But then $d+d = 2d \ne 1/3$.
\end{example}

\begin{theorem} \label{4.5} There exist density matrices
$\sigma\in M_3$ with eigenvalues $\lambda_1\ge \lambda_2\ge \lambda_3$ and $\rho\in M_6 $
with eigenvalues $\mu_1\ge \cdots\ge \mu_6$ such that $\tr_1(\rho) =\sigma$
if and only if $\mu_4+\mu_5\le\lambda_1\le  \mu_1+\mu_2$ and
$ \mu_5+ \mu_6\le \lambda_3\le  \mu_2+  \mu_3$.
\end{theorem}

\it Proof. \rm
Suppose $\rho=\(\begin{array}{cc}\rho_{11}&\rho_{12}\\
\rho_{21}&\rho_{22}\end{array}\)$ has  eigenvalues $\mu_1\ge \cdots\ge \mu_6$ such that
$\tr_1(\rho) =\sigma$. Then we may assume that $\rho_{11}+\rho_{22}
=\diag(\lambda_1, \lambda_2, \lambda_3)$. As in the proof of Theorem \ref{mn1}, we have
$(\lambda_1, \lambda_2, \lambda_3)\prec (\mu_1+\mu_2,\mu_3+\mu_4, \mu_5+ \mu_6)$.
Therefore,
we have $\lambda_1\le  \mu_1+\mu_2$ and $ \mu_5+ \mu_6\le \lambda_3$. Suppose $\rho_{11}$
and $\rho_{22}$ have eigenvalues $a_1\ge a_2\ge a_3$ and $b_1\ge b_2\ge b_3$ respectively.
Then applying  Horn inequalities \cite{Fu,Fu1} for the triple $((1,3),(1,3),(2,3))$, we have
$a_1+a_3+b_1+b_3\ge \lambda_2+\lambda_3\Ra \lambda_1\ge a_2 +b_2$. Let
$a_i=b_i=0$ for $i=4,5,6$.
By a result in \cite{LP03},  there exist $A,\ B\in H_6 $ with eigenvalues
$a_1\ge \cdots\ge a_6$ and $b_1\ge \cdots\ge b_6$ respectively, such that $A+B$ has eigenvalues
$\mu_1, \dots, \mu_6$.  Applying  Horn inequalities for the triple $((2,4),(2,4),(4,5))$,
we have $$\lambda_1\ge a_2 +b_2=a_2+a_4 +b_2+b_4\ge \mu_4+\mu_5$$
The inequality $\lambda_3\le  \mu_2+  \mu_3$ follows from symmetry.

Conversely, suppose $\mu_4+\mu_5\le\lambda_1\le  \mu_1+\mu_2$ and
$ \mu_5+ \mu_6\le \lambda_3\le  \mu_2+  \mu_3$. Then $\lambda_1$
lies in (at least) one of the following intervals:
\begin{equation}\label{mu}
  [\mu_5 + \mu_4 ,\mu_5 + \mu_3 ],
\ [\mu_5 + \mu_3 ,\mu_5 + \mu_2 ],
  [\mu_5 + \mu_2 ,\mu_4 + \mu_2 ],\ [\mu_4 + \mu_2 ,\mu_3 + \mu_2 ],
\ [\mu_3 + \mu_2 ,\mu_1 + \mu_2 ]
\end{equation}
Suppose $\mu_i+\mu_j\le \lambda_1\le \mu_i+\mu_k$. Then we can choose
$\mu_j\le \hat \mu_j,\ \hat\mu_k\le \mu_k$ such that $\mu_i+\hat\mu_j= \lambda_1$ and
$\mu_j+\mu_k=\hat\mu_j+\hat\mu_k$. Let $a=\sqrt{\hat\mu_j\hat\mu_k-\mu_j\mu_k}$.  Then
$\[\begin{array}{cc}\hat\mu_j&a\\ a&\hat \mu_k\end{array}\]$ has eigenvalues $\mu_j,\ \mu_k$.

Let the remaining 3 eigenvalues of $\rho$ be
$\{\mu_{i},\  \mu_{j},\ \mu_k\}^c=\{\mu_{i_1},\ \mu_{i_2},\ \mu_{i_3}\}$.

\medskip\noindent
{\bf Claim}  For some $\ell=2$ or $3$,  $\lambda_\ell$  satisfies i)
$\mu_{i_1}+\mu_{i_2}\le \lambda_{\ell}\le \mu_{i_1}+\mu_{i_3}$ or ii)
$\hat\mu_k+\mu_{i_2}\le \lambda_{\ell}\le \hat\mu_k+\mu_{i_3}$

Suppose i) in  the claim holds. Then we can choose
$\mu_{i_2}\le \hat \mu_{i_2},\ \hat\mu_{i_3}\le \mu_{i_3}$
such that $\mu_{i_1}+\hat\mu_{i_2}= \lambda_\ell$ and
$\mu_{i_2}+\mu_{i_3}=\hat\mu_{i_2}+\hat\mu_{i_3}$.
Let $b=\sqrt{\hat\mu_{i_2}\hat\mu_{i_3}-\mu_{i_2}\mu_{i_3}}$.
Then $\[\begin{array}{cc}\hat \mu_{i_2}&b\\ b&\hat\mu_{i_3}\end{array}\]$ has eigenvalues
$\mu_{i_2},\ \mu_{i_3}$. Hence, the matrix $$\rho=\[\begin{array}{ccc|ccc}\mu_i&0&0&0&0&0\\
0&\hat\mu_k&0&a&0&0\\
0&0&\hat\mu_{i_2}&0&b&0\\
\hline
0&a&0&\hat\mu_j&0&0\\
0&0&b&0&\hat\mu_{i_3}&0\\
0&0&0&0&0&\mu_{i_1}\end{array}\]$$
has eigenvalues $\mu_1,\dots,\mu_6$ and $\tr_1(\rho)$ has eigenvalues $\lambda_1,\ \lambda_2$
and $\lambda_3$.

The proof for the case ii) is similar.

We are going to show that the claim holds in each of the cases in (\ref{mu}).

\begin{enumerate}
\item $\mu_5 + \mu_4\le \lambda_1\le\mu_5 + \mu_3:$ Then the remaining 3 eigenvalues are
$\mu_1,\ \mu_2$ and $\mu_6$. We have
$\lambda_2\le \lambda_1\le\mu_5 + \mu_3\le \mu_2 + \mu_1$ and

$$\lambda_2\ge\sum_{i=1}^6\mu_i-2\lambda_1 \ge\sum_{i=1}^6\mu_i-2(\mu_5 + \mu_3)\ge \mu_2
+ \mu_6$$
\item $\mu_5 + \mu_3\le \lambda_1\le\mu_5 + \mu_2:$ Then the remaining 3 eigenvalues are
$\mu_1,\ \mu_4$ and $\mu_6$. We have
$\lambda_2\le \lambda_1\le\mu_2 + \mu_5\le \mu_1 + \mu_4$ and
$$\lambda_2\ge\(\sum_{i=1}^6\mu_i-\lambda_1\)/2 \ge\(\sum_{i=1}^6\mu_i-(\mu_5 + \mu_2)\)/2
\ge \mu_4 + \mu_6$$

\item $\mu_5 + \mu_2\le \lambda_1\le\mu_4 + \mu_2:$ Then the remaining 3 eigenvalues are
$\mu_1,\ \mu_3$ and $\mu_6$. We have
$\lambda_2\le \lambda_1\le\mu_2 + \mu_4\le \mu_1 + \mu_3$ and
$$\lambda_2\ge\(\sum_{i=1}^6\mu_i-\lambda_1\)/2 \ge\(\sum_{i=1}^6\mu_i-(\mu_4 + \mu_2)\)/2
\ge \mu_3 + \mu_6$$

\item $\mu_4 + \mu_2\le \lambda_1\le\mu_3 + \mu_2:$ Then the remaining 3 eigenvalues are
$\mu_1,\ \mu_5$ and $\mu_6$. We have $\lambda_3\ge   \mu_5 + \mu_6$ and
$$\lambda_3\le\(\sum_{i=1}^6\mu_i-\lambda_1\)/2 \le\(\sum_{i=1}^6\mu_i-(\mu_4 + \mu_2)\)/2 \le \mu_1 + \mu_5$$

\item $\mu_3 + \mu_2\le \lambda_1\le\mu_1 + \mu_2:$
Since $\mu_5 + \mu_6\le \lambda_3\le \mu_2+\mu_3$, consider the following cases:

\begin{enumerate}
\item
If  $\mu_5 + \mu_6\le \lambda_3\le   \mu_5 + \mu_4$, then we are done.

\item
If  $\mu_5 + \mu_4\le \lambda_3\le   \mu_3 + \mu_4$, then we use $\mu_6 + \mu_2
\le \lambda_1\le\mu_1 + \mu_2$ and we are done.

\item If  $\mu_3 + \mu_4\le \lambda_3\le   \mu_3 + \mu_2$, then we have
$$\lambda_1+\lambda_2\ge \mu_3 + \mu_2+\mu_3 + \mu_4\ge \mu_2+\mu_3 + \mu_4+\mu_6\Ra \lambda_2
\le \mu_1+\mu_5$$
So we can use  $\mu_6 + \mu_5\le \lambda_2\le\mu_1 + \mu_5$ and we are done. \qed
\end{enumerate}

\end{enumerate}

\section{Final remarks and further research}

First, let us give the solutions of the simple questions mentioned in Section 1
using the results in Section 3 and 4. (Theorems \ref{4.5}, \ref{pro:rank1}, and \ref{extreme}).

\begin{enumerate}
\item There exists a density matrix $\rho \in M_{2\cdot 3}$
with eigenvalues $a_1 \ge \cdots \ge a_6$
such that $\tr_1(\rho) = I_3/3$ if and only if
$$a_2+a_3 \ge 1/3 \ge a_4+a_5.$$
\item
There exists a density matrix in $\cS(I_3/3)$
with rank $k$ if and only if $2 \le k \le 6$.

\item
 There exists an extreme point of
 $\cS(I_3/3)$ with rank $k$ if and only if $2 \le k \le 3$.
\end{enumerate}

One may consider extending the results in the previous sections to
the compact convex set
$$\cS(\sigma_1,\sigma_2) = \{\rho \in D({m\cdot n}): \tr_1(\rho) = \sigma_2, \tr_2(\rho) = \sigma_1\}$$
for given $\sigma_1 \in D(m), \sigma_2 \in D(n)$.
As mentioned in the introduction, Klyachko \cite{Kl} has studied the relations between
the eigenvalues of $\rho\in \cS(\sigma_1,\sigma_2)$ and those of  $\sigma_1, \sigma_2$.
The answers depend on numerous linear inequalities that are difficult to handle.
As mentioned in the introduction, it is not easy
to generate and store all the inequalities  and
it is hard to use them to deduce answers for simple problems such as:

\medskip
\begin{problem}
Determine the ranks of the elements in $\cS(\sigma_1,\sigma_2)$.
\end{problem}

\begin{problem}
Determine the ranks of the extreme points of the set $\cS(\sigma_1,\sigma_2)$.
\end{problem}

\medskip
\noindent
Note also that unlike the case of $\cS(\sigma)$, the ranks of the elements
in $\cS(\sigma_1, \sigma_2)$ cannot be determined only by the ranks of $\sigma_1$ and
$\sigma_2$. For example, suppose $\sigma_1$ and $\sigma_2$ have the same rank.
If $\sigma_1, \sigma_2$ have the same set of non-zero
eigenvalues, then there is a rank one matrix in $\cS(\sigma_1,\sigma_2)$. Otherwise,
there is no rank one matrix in $\cS(\sigma_1,\sigma_2)$.
While it is difficult to determine the minimum rank of the matrices
in $\cS(\sigma_1,\sigma_2)$,  it is easy to show that the
largest rank of the matrices in $\cS(\sigma_1,\sigma_2)$ equal
$\rank(\sigma_1)\rank(\sigma_2)$. Also, it is not hard to show that a matrix
in $\cS(\sigma_1,\sigma_2)$ with minimum rank is an extreme point. However, it is not
easy to determine the ranks of extreme points in general.
In \cite{Pa}, it was shown that the rank of an extreme point in
$\cS(\sigma_1,\sigma_2)$ cannot exceed $(m^2+n^2-1)^{1/2}$.
In fact, one can
show that if $\sigma_i$ has rank $r_i$ for $i = 1,2$,
then the rank of an extreme point of $\cS(\sigma_1,\sigma_2)$ cannot exceed
$(r_1^2+r_2^2-1)^{1/2}$ based on the following extension of
Lemma \ref{lem1}.

\begin{lemma} \label{lem1-e}
 Let $\sigma_1\in D(m),\ \sigma_2\in D(n)$ and  $\rho \in
\cS(\sigma_1,\sigma_2) \subseteq D(m\cdot n)$. Then
$\rho$ is not an extreme point if and only if
there exists a nonzero $\xi\in H_{mn}$ such that
$\rho \pm \xi \in PSD(m\cdot n)$,  $\tr_1(\xi) = O_n$ and $\tr_2(\xi) = O_m$.
In such a case, there are $\rho_1, \rho_2 \in \cS(\sigma_1,\sigma_2)$
with $\rank(\rho_1) < \rank(\rho)$ such that
$\rho = (\rho_1+\rho_2)/2$.
\end{lemma}

Of course, similar questions can be asked for the set
$$\cS(\sigma_1, \dots, \sigma_k) = \{\rho\in D({n_1\cdots n_k}): \tr_{j'}(\rho) = \sigma_j\},$$
where $\tr_{j'}(\rho)$ is the reduced state of $\rho$ in the $j$th system,
for given $\sigma_j \in D({n_j})$ with $j = 1, \dots, k$.
Even more challenging problems will be the study of $\rho$ and
reduced states in subsystems that have overlaps.

\bigskip\noindent
{\bf Acknowledgment}

The research of Li and Poon
was supported by USA NSF, and HK RGC.
Li  was an honorary professor of the Shanghai University,
and an honorary professor of the University of Hong Kong.
The research of Wang was done while he was visiting the College of William and
Mary during the academic
year 2013-14 under the support of China Scholarship Council. We would like to thank the referee for some helpful comments.

\bibliographystyle{amsplain}

\begin{thebibliography}{www}



\bibitem{Cet}
J. Chen, Z. Ji, D.W. Kribs, A. Klyachko, B. Zeng,
Rank reduction for the local consistency problem: an algebraic geometry approach,
Journal of Mathematical Physics 53 (2012), 022202.


\bibitem{Hy}
S. Daftuar and P. Hayden,
Quantum state transformations and the Schubert calculus,
Annals of Physics 315 (2005), 80-122.

\bibitem{FP} K. Fan and G. Pall, Imbedding conditions for Hermitian and
normal matrices,   Canad. J. Math.  9 (1957), 298-304.


\bibitem{Fu} W. Fulton,
Eigenvalues, invariant factors, highest weights, and Schubert calculus,
  Bull. Amer. Math. Soc.    37 (2000), 209--249.


\bibitem{Fu1} W. Fulton, Eigenvalues of majorized Hermitian matrices,
and Littlewood-Richardson coefficients,  Linear Algebra Appl.
319 (2000), 23-36.


\bibitem{H}
A. Horn,  Doubly stochastic matrices and the diagonal of a rotation matrix,
Amer. J. Math. 76  (1954), 620--630.


\bibitem{Kl} A. Klyachko,   Quantum marginal problem and N-representability,
J. Phys. Conf. Series 36 (2006), 72–-86.

\bibitem{K} K. Kraus, States,  effects, and operations : fundamental
notions of quantum theory, Lectures in mathematical physics at the University
of Texas at Austin, Lecture Notes in Physics 190, Springer-Verlag,
Berlin-Heidelberg, 1983.

\bibitem{LP03} C.K. Li and Y. T. Poon, Principal Submatrices of a Hermitian
matrix,   Linear and Multilinear Algebra  51 (2003), 199-208.

\bibitem{LT} C.K. Li and N.K. Tsing,
Norms that are invariant under unitary similarities and the $C$-numerical radii,
Linear and Multilinear Algebra 24 (1989), 209-222.

\bibitem{MO} A.W. Marshall and I. Olkin,   Inequalities: The Theory of
Majorizations and Its Applications,   Academic Press, 1979.

\bibitem{NC} M.A. Nielsen and I.L. Chuang, Quantum Computation and
Quantum information, Cambridge University Press, 2000.


\bibitem{Pa} K.R. Parthasarathy,
Extremal quantum states in coupled systems,
Ann. I. H. Poincar\'{e} -- PR 41 (2005), 257--268.


\bibitem{Ru} O. Rudolph,
On extremal quantum states of composite systems with
fixed marginals, Journal of Mathematical Physics 45 (2004), 4035-4041.


\end{thebibliography}


\end{document}